\documentclass[a4paper,USenglish]{article}

\usepackage{amssymb}
\usepackage{babel}
\usepackage[tbtags,fleqn]{amsmath}
\usepackage{amsthm}
\theoremstyle{plain}
\newtheorem{theorem}{Theorem}

\theoremstyle{definition}
\newtheorem{definition}[theorem]{Definition}

\theoremstyle{remark}

\usepackage{caption}
\RequirePackage[figuresright]{rotating}
\RequirePackage{subfig}

\usepackage{authblk}

\title{A CEGAR-like Approach for Cost LTL Bounds}

\date{}
\author[1]{Maximilien Colange}
\author[1]{Dimitri Racordon}
\author[1]{Didier Buchs}
\affil[1]{Centre Universitaire d'Informatique -- University of Geneva\\ \texttt{first.last@unige.ch}}

\usepackage[nocompress]{cite}

\usepackage{dsfont}

\usepackage{paralist}

\usepackage{tikz}

\usepackage{stmaryrd}

\usepackage[linesnumbered,procnumbered,ruled,vlined]{algorithm2e}

\usetikzlibrary{arrows, automata, backgrounds, petri, positioning}

\usepackage{wrapfig}

\newcommand{\setof}[1]{\ensuremath{\left \{ #1 \right \}}}
\newcommand{\tuple}[1]{\ensuremath{\left \langle #1 \right \rangle }}

\newcommand{\ltlg}{\ensuremath{\operatorname{\mathbf{G}}}}
\newcommand{\ltlf}{\ensuremath{\operatorname{\mathbf{F}}}}
\newcommand{\ltlx}{\ensuremath{\operatorname{\mathbf{X}}}}
\newcommand{\ltlu}{\ensuremath{\operatorname{\mathbf{U}}}}
\newcommand{\ltlr}{\ensuremath{\operatorname{\mathbf{R}}}}

\newcommand{\ltl}{\text{LTL}}
\newcommand{\cltl}{\text{LTL$^{\leq}$}}
\newcommand{\cltlbar}{\text{LTL$^{>}$}}

\newcommand{\sem}[1]{\ensuremath{\llbracket #1 \rrbracket_{\leq}}}
\newcommand{\sembar}[1]{\ensuremath{\llbracket #1 \rrbracket_{>}}}

\newcommand{\modelsinf}[0]{\ensuremath{\models_{\leq}}}
\newcommand{\modelssup}[0]{\ensuremath{\models_{>}}}

\newcommand{\lang}[1]{\ensuremath{\mathcal{L}(#1)}}

\newcommand{\increment}{\ensuremath{\mathtt{i}}}
\newcommand{\reset}{\ensuremath{\mathtt{r}}}
\newcommand{\observe}{\ensuremath{\mathtt{o}}}
\newcommand{\checkreset}{\ensuremath{\observe{}\reset{}}}

\theoremstyle{plain}
\newtheorem{proposition}[theorem]{Property}

\tikzset{
  bacc/.style={circle,fill=black,draw=black,scale=0.5}
}

\begin{document}

\maketitle

\begin{abstract}
Qualitative formal verification, that seeks boolean answers about the behavior of a system, is often insufficient for practical purposes.
Observing quantitative information is of interest, e.g. for the proper calibration of a battery or a real-time scheduler.
Historically, the focus has been on quantities in a continuous domain, but recent years showed a renewed interest for discrete quantitative domains.

Cost Linear Temporal Logic (CLTL) is a quantitative extension of classical LTL.
It integrates into a nice theory developed in the past few years that extends the qualitative setting, with counterparts in terms of logics, automata and algebraic structure.
We propose a practical usage of this logics for model-checking purposes.
A CLTL formula defines a function from infinite words to integers.
Finding the bounds of such a function over a given set of words can be seen as an extension of LTL universal and existential model-checking.
We propose a CEGAR-like algorithm to find these bounds by relying on classical LTL model-checking, and use B\"{u}chi automata with counters to implement it.
This method constitutes a first step towards the practical use of such a discrete quantitative logic. 
\end{abstract}

\section{Introduction}

Qualitative verification, asking questions with boolean answers about a system may be too strict for various applications.
Calibrating a battery, timing a scheduler, measuring quality of service are practical problems of systems designers for which formal verification can offer a guarantee.
Many works focus on the case of continuous quantitative domains (stochastic systems, real-time systems \dots), and the case of discrete domains have long been overlooked.

The ability to count events is an important feature, e.g. to evaluate logical time (number of actions done by a robot, number of context switches done by a scheduler \dots).
Such measurements are of primary interest to evaluate the behavior of a system at early stages of development.
Logical time can also serve as a first approximation of real-time, when events have a known bounded duration.
We seek in this paper to use a logic able to count events in a system with infinite behaviors, with a focus on applicability.
Following the automata-approach largely adopted for Linear Temporal Logic (\ltl{}) verification, we study a quantitative extension of automata able to count events.

Among numerous quantitative extensions to finite automata, we focus on the one defined by Colcombet and Boja\'{n}czyk~\cite{bojcol:2006}.
Finite automata are extended with a finite set of counters that can be incremented and reset.
A special operation \emph{observe} allows to store the current value of a counter to further determine the value of a run, as the infimum or the supremum of such stored values.
They are part of a vast theory that nicely extends the finite automata theory, with their logical and algebraic counterparts, closure properties, over finite and infinite words and finite trees.
Such automata define functions from words to integers, termed \emph{cost functions}.
Due to the undecidability of comparing two cost functions, many nice features of this theory rely on the consideration of cost functions up to an equivalence relation that “erases” the exact values and only retains their boundedness of functions.

Regarding infinite words, on which we focus in this paper, the theory of cost functions has links with other extensions of automata or logics, by bounding a discrete quantity:
bounding the maximal time between two returns to an accepting state in $\omega$-automata~\cite{almagor2010promptness}, or bounding the wait time for the \emph{finally} operator in LTL~\cite{kupferman2007liveness}.
Considering exact values to count events is nevertheless a great tool for verification.
Think for example to the maximal number of energy units consumed by a robot between two returns to its charging base, to calibrate a battery.
Or the maximal number of simultaneous threads in a parallel computation, to tune an appropriate scheduler.
Or the number of false steps permitted to a human operator before a safeguard restriction occurs.
For such properties, determining whether the bound is finite or not is of little help.
We thus propose to use the tools and methods developed towards cost function theory (over infinite words) to practical model-checking.

 We use a counting extension of LTL, Cost LTL, introduced in~\cite{kuperberg2012expressive} in the context of the theory of cost functions.
Our contribution is threefold:
\begin{compactitem}
  \item Cost LTL being an extension of LTL, we show how classical model-checking problems on LTL extend to the quantitative case.
  We thus address the problem of finding the bounds of the cost function defined by a formula, with a focus on the upper bound search.
  We propose an algorithm to compute such an upper bound using a CEGAR-like approach, where the bound is computed thanks to successive refinements.
  \item We show how this algorithm is effective, implementing it by means of $\omega$-automata with counters.
  The computational bottleneck of the algorithm is reduced to B\"{u}chi automata emptiness check, to take advantage of existing research in the field. 
  \item We also present concrete examples of application of Cost LTL, to illustrate its potential as a practical tool for verification.
\end{compactitem}

The paper is organized as follows:
Section~\ref{sec:cltl} first presents Cost LTL, introduced in~\cite{kuperberg2012expressive}, and some basic results used in the remainder of the paper.
The core of our contribution is a CEGAR-like algorithm to determine the bounds of a cost function defined by a CLTL property, in Section~\ref{sec:cegar}.
We then show in Section~\ref{sec:automata} how this algorithm can be effectively implemented thanks to $\omega$-automata equipped with counters.
Finally, Section~\ref{sec:related} presents related work, and Section~\ref{sec:conclusion} concludes our study and proposes leads for future developments. 

\subparagraph*{Notations}
Given $u \in \Sigma^{\omega}$ an infinite word over an alphabet $\Sigma$,
and $i \in \mathds{N}$, $u_i$ is the $i$-th letter of $u$, and $u^i$ the suffix of $u$ starting at $u_i$.
Thus $u = u^0$ and $u = u_0 \dots u_{i-1} u^i$ for any $i > 0$.
If $A$ is a finite set, $|A|$ denotes its cardinal.
For $A \subseteq \mathds{N}$, $\inf A$ (resp. $\sup A$) denotes the infimum (resp. supremum) of $A$.
By convention, $\inf \emptyset = + \infty$ and $\sup \emptyset = 0$.
Let a $f \in \mathds{N}^D$ for some set $D$.
For $D' \subseteq D$, the image of $D'$ by $f$ is the set $f(D') = \{ f(x) ~|~ x \in D' \}$.
We also note $\sup_{D'} f = \sup f(D')$, respectively $\inf_{D'} f = \inf f(D')$.

\section{Cost Linear Temporal Logics}
\label{sec:cltl}

We first define Cost Linear Temporal Logic (\cltl{}), as in~\cite{kuperberg2012expressive}.
Let $AP$ be a set of atomic propositions.
The set of \cltl{} formulae is defined by ($a$ ranges over $AP$):
$$\phi ::= a ~|~ \neg a ~|~ \phi \vee \phi ~|~ \phi \wedge \phi ~|~ \phi \ltlu{} \phi ~|~ \phi \ltlr{} \phi ~|~ \ltlx{} \phi ~|~ \phi \ltlu{}^{\leq} \phi$$

Every \ltl{} formula has a semantically equivalent formula in Negative Normal Form (NNF), where negations can only appear in front of an atomic proposition.
Any \ltl{} formula in NNF is a \cltl{} formula, and in that sense, \ltl{} is a strict subset of \cltl{}.
From now on, we identify \ltl{} with \ltl{} in NNF, so that $\ltl{} = \cltl{} - \{ \ltlu{}^{\leq} \}$ and $\ltl{} \subsetneq \cltl{}$.

A formula of \cltl{} is evaluated over infinite words on the alphabet $2^{AP}$.
Let $u \in (2^{AP})^{\omega}$, $n \in \mathds{N}$, $\phi_1$, $\phi_2$ be \cltl{} formulae, and $a \in AP$:
\begin{align*}
&(u,n) \modelsinf a                      &&\text{ iff } a \in u_0 \\
&(u,n) \modelsinf \neg a                 &&\text{ iff } a \notin u_0 \\
&(u,n) \modelsinf \phi_1 \vee \phi_2     &&\text{ iff } (u,n) \modelsinf \phi_1 \text{ or } (u,n) \modelsinf \phi_2         \\
&(u,n) \modelsinf \phi_1 \wedge \phi_2   &&\text{ iff } (u,n) \modelsinf \phi_1 \text{ and } (u,n) \modelsinf \phi_2        \\
&(u,n) \modelsinf \ltlx{} \phi_1         &&\text{ iff } (u^1,n) \modelsinf \phi_1                                           \\
&(u,n) \modelsinf \phi_1 \ltlu{} \phi_2  &&\text{ iff } \exists i \in \mathds{N} \text{ s.t. } (u^i,n) \modelsinf \phi_2    \\
&                                        &&\text{ and } \forall j < i, (u^j,n) \modelsinf \phi_1                            \\
&(u,n) \modelsinf \phi_1 \ltlr{} \phi_2  &&\text{ iff } \forall i \in \mathds{N} \text{ either } (u^i,n) \modelsinf \phi_2  \\
&                                        &&\text{ or } \exists j < i \text{ s.t. } (u^j,n) \modelsinf \phi_1                \\
&(u,n) \modelsinf \phi_1 \ltlu{}^{\leq} \phi_2 &&\text{ iff } \exists i \in \mathds{N} \text{ s.t. } (u^i,n) \modelsinf \phi_2 \\
&                                        &&\text{ and } |\{ j < i ~|~ (u^j,n) \not\modelsinf \phi_1 \}| \leq n
\end{align*}

The semantics of $\phi \in \cltl{}$ is the function\\
$$\sem{\phi}: \left. \begin{tabular}{cl}$(2^{AP})^{\omega}$ &$\mapsto \mathds{N} \cup \{ \infty \}$ \\ $u$ &$\mapsto \inf~\{ n ~|~ (u,n) \modelsinf \phi \}$ \end{tabular}\right.$$

To keep examples clear, we identify any atomic proposition $a$ with the subset of $2^{AP}$ of sets containing $a$, $\bar{a}$ being its complementary.
Consider the formula $\phi_1 = \ltlf{}^{\leq} \neg a$, short for $\bot \ltlu{}^{\leq} \neg a$.
For any $n \leq p$ and any word $u \in a^n \bar{a} (2^{AP})^\omega$, $(u, p) \modelsinf \phi_1$, and $\sem{\phi_1}(u) = n$.
Consider now the formula $\phi_2 = \ltlg{} (\ltlf{}^{\leq} \neg a)$, short for $\bot \ltlr{} \ltlf{}^{\leq} \neg a$.
For $n \in \mathds{N}$ and $u \in (2^{AP})^{\omega}$, $(u,n) \modelsinf \phi_2$ only if the distance between a letter in $\bar{a}$ and the next one also in $\bar{a}$ never exceeds $n$.
Thus $\sem{\phi_2}(u)$ is the maximal number of consecutive $a$'s in $u$.

If $\phi$ is a \ltl{} formula, either $(u,n) \modelsinf \phi$ holds for every $n$, in which case $\sem{\phi}(u) = 0$, or for none, in which case $\sem{\phi}(u) = \infty$.
The former is noted $u \vdash \phi$, and matches the usual semantics of \ltl{}.
In other words, the value \texttt{true} is mapped onto $0$ and \texttt{false} onto $\infty$.

From the semantical definition, for any integers $n \leq p$, if $(u,n) \modelsinf \phi$ then $(u,p) \modelsinf \phi$ too.
Stating $(u,n) \modelsinf \phi$ is thus equivalent to stating $\sem{\phi}(u) \leq n$.

For $n \in \mathds{N}$, we propose a translation from a \cltl{} formula $\phi$ to a \ltl{} formula $\phi[n]$ that separates words according to their value relatively to $n$.
More precisely, $u \vdash \phi[n]$ if, and only if, $\sem{\phi}(u) \leq n$.
$\phi[n]$ is defined inductively as follows ($\phi_1, \phi_2 \in \cltl{}$ and $a \in AP$):
\begin{compactitem}
  \item $a[n] = a$ and $(\neg a)[n] = \neg a$;
  \item $(\ltlx{} \phi_1)[n] = \ltlx{}(\phi_1[n])$ 
  \item $(\phi_1 \bowtie \phi_2)[n] = \phi_1[n] \bowtie \phi_2[n]$ for $\bowtie \in \{ \vee, \wedge, \ltlu{}, \ltlr{} \}$;
  \item for $\phi_1, \phi_2 \in \ltl{}$, $(\phi_1 \ltlu{}^{\leq} \phi_2)[0] = \phi_1 \ltlu{} \phi_2$\\and $(\phi_1 \ltlu{}^{\leq} \phi_2)[n+1] = (\phi_1 \vee \ltlx{}(\phi_1 \ltlu{}^{\leq} \phi_2)[n]) \ltlu{} \phi_2$;
  \item otherwise $(\phi_1 \ltlu{}^{\leq} \phi_2)[n] = (\phi_1[n] \ltlu{}^{\leq} \phi_2[n])[n]$.
\end{compactitem}

Back to our example $\phi_1 = \ltlf{}^{\leq} \neg a$, we have $\phi_1[0] = \bot \ltlu{}  \neg a = \neg a$, hence $\phi_1[1] = (\ltlx{} \neg a) \ltlu{} \neg a$, equivalent to $\neg a \vee \ltlx{} \neg a$.
Thus, $\phi_1[n] = \vee_{i=0}^n X^i \neg a$ for every $n$.

\begin{proposition}
\label{prop:unfolding}
For $u \in (2^{AP})^{\omega}$, $n \in \mathds{N}$ and $\phi \in \cltl{}$, $u \vdash \phi[n]$ iff $\sem{\phi}(u) \leq n$.
\end{proposition}
\begin{proof}
Structural induction on $\phi$, detailed in appendix.
\end{proof}

\subsection{Dual Logics}

In \cltl{}, negations can only appear in the leaves of the formula, so that a formula is always in NNF.
This particularity is commanded by the semantical difficulty to negate the operator $\ltlu{}^{\leq}$.
In a boolean setting, a word is either a model of the formula, or it is not.
In our quantitative setting, negation is not straightforward, as it is not a natural operation over $\mathds{N}$.
We take inspiration from the embedment of \ltl{} in \cltl{}: \texttt{true} corresponds to $0$ and \texttt{false} to $\infty$.
Semantically, the negation thus replaces $\inf$ with $\sup$.
We define the logic \cltlbar{}~~\cite{kuperberg2014linear}, dual to \cltl{}.
The operator $\ltlu{}^{\leq}$ is replaced by $\ltlr{}^{>}$ whose semantics is defined so as to match the negation of $\ltlu{}^{\leq}$ semantics:
$(u,n) \modelssup \phi_1 \ltlr{}^{>} \phi_2$ iff for every $i \in \mathds{N}$, either $(u^i,n) \modelssup \phi_2$ or $|\{ j < i ~|~ (u^j,n) \modelssup \phi_1 \}| > n$.
All other operators keep their natural semantics.
The semantics of $\phi \in \cltlbar{}$ is a function 
$$\sembar{\phi}: \left. \begin{tabular}{cl}$(2^{AP})^{\omega}$ &$\mapsto \mathds{N} \cup \{ \infty \}$ \\ $u$ &$\mapsto \sup~\{ n ~|~ (u,n) \modelssup \phi \}$ \end{tabular}\right.$$
For any $n \leq p$, if $(u,p) \modelssup \phi$ then $(u,n) \modelssup \phi$, so that $(u,n) \modelssup \phi$ iff $\sembar{\phi}(u) \geq n$.

Note that \ltl{} is also embedded in \cltlbar{}, with a semantics dual to the case of \cltl{}: \texttt{true} is mapped onto $\infty$ and \texttt{false} onto $0$.
Syntactically, $\ltl{} = \cltl{} \cap \cltlbar{}$, but the semantics do not match.
Note that in both cases (\cltl{} and \cltlbar{}), either $(u,n)$ is a model for $\phi \in \ltl{}$ for every $n \in \mathds{N}$, or for none.
From now on, we note $u \vdash \phi$ for the former case, to be matched with the appropriate semantics depending on context.

Syntactically, we get dual pairs of operators: $\vee$ and $\wedge$, $\ltlu{}$ and $\ltlr{}$, $\ltlu{}^{\leq}$ and $\ltlr{}^{>}$, $\ltlx{}$ being self-dual.
From a formula $\phi \in \cltl{}$ (resp. \cltlbar{}), we can build a formula $\neg \phi \in \cltlbar{}$ (resp. \cltl{}), by pushing the negation to the leaves: the top operator is replaced by its dual, and the negation is recursively pushed to the leaves of its operands.
Literals respect the excluded middle, which is semantically consistent, so as to eliminate double negations.
Syntactically, the excluded middle also holds, as pushing negations to the leaves in $\neg \neg \phi$ yields $\phi$.
Observe how the semantics of $\phi$ and $\neg \phi$ are correlated:

\begin{proposition}
For $u \in (2^{AP})^{\omega}$ and $\phi \in \cltl{}$, $\sembar{\neg \phi}(u) = \max (0, \sem{\phi}(u)-1)$.
\end{proposition}
\begin{proof}
An easy induction on $\phi$ proves that $(u,n) \modelsinf \phi$ iff $(u,n) \not\modelssup \neg \phi$ (observe how the semantics of dual operators are dual to each other).
We recall that $(u,n) \modelsinf \phi \implies (u,p) \modelsinf \phi$ for any $n \leq p$.
Thus $(u,n) \modelssup \neg \phi$ if, and only if, $n < \sem{\phi}(u) = \inf~\{ p ~|~ (u,p) \modelsinf \phi \}$.
Furthermore, if $\sem{\phi}(u) = 0$, then $\sembar{\phi}(u) = \sup~\emptyset = 0$ by convention.
%
\end{proof}

\begin{proposition}
For any $u \in (2^{AP})^{\omega}$ and $\phi \in \cltlbar{}$, $$\sem{\neg \phi}(u) = \begin{cases}
\sembar{\phi}(u) + 1  ~~ \text{if } \sembar{\phi}(u) > 0 \\
0                     ~~ \text{if } \forall n \in \mathds{N}. (u,n) \not\modelssup \phi \\
1                     ~~ \text{otherwise}
\end{cases}$$
\end{proposition}
\begin{proof}
Recall from above that $(u,n) \modelssup \phi$ if, and only if, $(u,n) \not\modelssup \neg \phi$.
If $\sembar{\phi}(u) > 0$, then $(u,n) \modelsinf \phi$ if, and only if, $n > \sembar{\phi}(u)$.
If $\sembar{\phi}(u) = 0$, $\{ n | (u,n) \modelssup \phi \}$ is either $\emptyset$ or $\{0\}$.
We get $\sem{\neg \phi}(u) = 0$ in the former case, and $\sem{\neg \phi}(u) = 1$ in the latter.
\end{proof}

Following the above, we define, for $\phi \in \cltlbar{}$ and $n \in \mathds{N}$, $\phi[n]$ as $\neg ((\neg \phi)[n])$.
\begin{proposition}
\label{prop:sup_instantiation}
For $u \in (2^{AP})^{\omega}$, $n \in \mathds{N}$ and $\phi \in \cltlbar{}$, $u \vdash \phi[n]$ iff $\sembar{\phi}(u) \geq n$.
\end{proposition}
\begin{proof}
$\sembar{\phi}(u) \geq n$ iff $(u,n) \modelssup \phi$ iff $(u,n) \not\modelsinf \neg \phi$ iff $u \not\vdash (\neg \phi)[n]$ iff $u \vdash \phi[n]$.
\end{proof}

\subsection{Cost Logics for Verification}

A common task in verification is whether a system has a behavior satisfying a given property.
The property either expresses a desired behavior, or an unwanted one (in which case finding a satisfying behavior amounts to finding a bug).
Typically, the behaviors of the system is a regular $\omega$-language $L$, and the property $\phi$ is a \ltl{} property.
Thanks to the closure properties of regular $\omega$-languages, this problem reduces to \emph{existential model-checking}: is the intersection of $L$ and the language recognized by $\phi$ empty?
\emph{Universal model-checking} asks whether a language $L'$ contains all $\omega$-words, and is dual to existential model-checking, since $L' = \Sigma^{\omega}$ iff $\Sigma^{\omega} - L' = \emptyset$.
As both \cltl{} and \cltlbar{} extend \ltl{}, the natural question we address is the extension of these two problems to the quantitative setting.

We first rephrase the \ltl{} existential and universal model-checking with the \cltl{} semantics: existential model-checking asks whether there is a word of value $0$.
Dually, universal model-checking asks whether all words have value $\infty$.
Existence of a word of value equal to, greater than, or less than a given $n$ are natural extensions of this question.
These questions hardly extend the boolean framework: by comparing word values against a \emph{given} $n$, they remain boolean questions.
We seek here a question with a \emph{quantitative} answer (in our case in the domain $\mathds{N} \cup \{ \infty \}$).
Two particular values of interest are the bounds of $\sem{\phi}$.

\begin{definition}
\emph{$\inf$-bound checking}: given $L \subseteq (2^{AP})^{\omega}$ and $\phi \in \cltl{}$, compute $\inf_L \sem{\phi}$.
\end{definition}

\begin{definition}
\emph{$\sup$-bound checking}:given $L \subseteq (2^{AP})^{\omega}$ and $\phi \in \cltl{}$, compute $\sup_L \sem{\phi}$.
\end{definition}

The duality of \cltl{} and \cltlbar{} allows to choose from $\phi$ or $\neg \phi$, as \cltlbar{} seems more appropriate for $\sup$-bound checking.
All the problems mentioned above are reducible to these two problems.
In \cltl{} semantics, existential \ltl{} model-checking boils down to $\inf$-bound checking, while its universal counter-part corresponds to $\sup$-bound checking.

\section{CLTL Bounds Checking}
\label{sec:cegar}

This section presents our main contribution, an algorithm to compute bounds for \cltl{} and \cltlbar{} formulae.
It is inspired by the Counter-Examples Guided Abstraction and Refinement (CEGAR) approach to qualitative model-checking, which we present first.

\subsection{The CEGAR Approach to qualitative model-checking}

Consider a regular $\omega$-language $L$ (the set of behaviors of a system) and a \ltl{} formula $\phi$.
We ask whether all words of $L$ are models of $\phi$.
It boils down to existential model-checking: is there a model of $\neg \phi$ in $L$?
The language of $\neg \phi$, $\lang{\neg \phi}$ is in deed regular, the intersection of $\lang{\neg \phi}$ and $L$ computable, and testing the emptiness of a regular $\omega$-language is decidable.
These steps are usually performed with $\omega$-automata to represent the regular $\omega$-languages.

But this approach becomes hardly tractable when the underlying automata are huge, as it is often the case when the input language is the set of behaviors of a concurrent system.
We present the so-called CEGAR (Counter-Example Guided Abstraction and Refinement) loop.
A language $L'$ larger (for the inclusion) than $L$ is called an \emph{abstraction}, $L$ being a \emph{refinement} of $L'$.
CEGAR loop assumes the existence of a refinement function $\rho$ that given a word $u \notin L$ and a regular abstraction $M$ of $L$, returns a regular refinement $M'$ of $M$ that does not contain $u$, and that is also an abstraction of $L$:
$L \subseteq M \text{ and } u \notin L \Rightarrow L \subseteq \rho(M,u) \subseteq M - \{u\}$.

The CEGAR loop proceeds as follows:
i) start from an abstraction $M$ of $L$;
ii) search in $M$ a model $u$ for $\neg \phi$;
iii) if there is no such $u$ in $M$, there is none in $L$ either, and the question is settled;
iv) otherwise, check whether $u \in L$;
v) if $u \in L$, then the question is settled;
vi) else start over with $M = \rho(M,u)$.
In practice, the automaton for $L$ is huge, and CEGAR avoids its full exploration by manipulating abstractions, that have smaller underlying automata.
Since a counter-example guides the refinement, the same spurious counter-example cannot be encountered twice.
In general, termination depends on the refinement function $\rho$, but practically, termination is easy to ensure, for example by falling back to the initial input $L$ (worst-case scenario) when the size of $M'$ exceeds the size of $L$.

\subsection{CEGAR Approach for Bounds Checking}

We adapt the CEGAR approach to solve $\sup$-bound checking for a \cltlbar{} formula $\phi$ over a regular $\omega$-language $L$.
The dual $\inf$-bound checking follows the same scheme.

We first have to define the notion of abstractions and refinements, thanks to an ordering over semantic functions: smaller elements are refinements and greater ones are abstractions.

%
%

\begin{definition}
\label{def:refinement}
For $L \subseteq \Sigma^{\omega}$, $f \preceq_L g$ iff $\sup_L f = \sup_L g$ and $g^{-1}(0) \subseteq f^{-1}(0)$.
\end{definition}

\SetKwFunction{ComputeBound}{ComputeBound}

\noindent
\begin{minipage}{0.5\textwidth}
\vspace{0pt}
\begin{algorithm}[H]
  $n \gets 0$ \;
  \While{true}{
    $\phi \gets \phi_0 \wedge \phi_0[n+1]$\label{alg:phi0} \;
    \eIf{$\exists~u \in L$ s.t. $\sembar{\phi}(u) > 0$\label{alg:keypoint}}{
      $n \gets p$ for any $n < p \leq \sembar{\phi}(u)$\label{alg:updaten} \;
    }{
      \Return $n$\;
    }
  }
\caption{ComputeBound($L$, $\phi_0$)\label{algo:computebound}}
\end{algorithm}
\end{minipage}
\begin{minipage}{0.5\textwidth}
In other words, $f$ refines $g$ (relatively to $L$) if they have the same supremum over $L$ and $f$ maps more words (for the inclusion) onto $0$ than $g$.
The CEGAR loop for our quantitative setting is shown in Algorithm~\ref{algo:computebound}.

We present the algorithm in its full generality.
Line~\ref{alg:updaten} leaves some room for various implementations, as we will see in Section~\ref{sec:cegar_implementation}.
\end{minipage}

The key of this algorithm is the search for a word $u$ such that $\sembar{\phi}(u) > 0$ for a \cltlbar{} formula $\phi$ (line~\ref{alg:keypoint}).
Note that if $\sup_L \sembar{\phi} > 0$, then for all $u \in L$, $\sembar{\phi}(u) > 0$ iff there is some $p$ such that $(u,p) \modelssup \phi$.
Considering the semantics of \cltlbar{}, this is equivalent to finding a word satisfying the \ltl{} formula $\phi'$, a copy of $\phi$ in which every occurrence of $\phi_1 \ltlr{}^{>} \phi_2$ is replaced by $\top \ltlr{} \phi_2$ (where $\top = a \vee \neg a$ for any $a \in AP$).
Therefore, the search for the upper bound of $\sembar{\phi}$ is reduced to \ltl{} emptiness check, a well-studied problem with numerous efficient solutions (see~\cite{rozier2010ltl, renault2013scc} for surveys).
The corner case $\sup_L \sembar{\phi} = 0$ can be detected at the second pass in the loop (see the proof of Proposition~\ref{prop:correctness}).

\begin{proposition}
\label{prop:cegar_equality}
At line~\ref{alg:keypoint}, for all $u \in L$, 
$\sembar{\phi}(u) =
\begin{cases}
  \sembar{\phi_0}(u)  &\text{if } \sembar{\phi_0}(u) > n \\
  0                   &\text{otherwise}
\end{cases}$
\end{proposition}
\begin{proof}
At line~\ref{alg:keypoint}, $\phi = \phi_0 \wedge \phi_0[n+1]$.
Let $u \in L$ and $m = \sembar{\phi_0}(u)$.
By Proposition~\ref{prop:sup_instantiation}, $m > n$ iff $(u,m) \modelssup \phi_0[n+1]$.
$\sembar{\phi}(u)$ is the largest $p$ such that both $(u,p) \modelssup \phi_0$ and $(u,p) \modelssup \phi_0[n+1]$.
If $m > n$, $m$ is the largest such $p$, so that $\sembar{\phi}(u) = m$.
Otherwise, there are no value complying to the latter condition, and $\sembar{\phi}(u) = \sup \emptyset = 0$.
\end{proof}

Proposition~\ref{prop:cegar_equality} proves that $\sembar{\phi} \preceq_L \sembar{\phi_0}$, and that at each pass in the loop, $\phi$ is refined with respect to $\preceq_L$.

\begin{proposition}
\label{prop:correctness}
If $\sup_{L} \sembar{\phi_0}$ is finite, \ComputeBound is both correct and sound, i.e. it terminates and returns $\sup_L \sembar{\phi_0}$.
\end{proposition}
\begin{proof}
At line~\ref{alg:updaten}, $n$ is updated with a value $p$ such that $n < p \leq \sembar{\phi_0}(u)$ for some $u \in L$ such that $\sembar{\phi}(u) > 0$.
Proposition~\ref{prop:cegar_equality} guarantees the existence of such a $p$, and $n$ strictly increases when updated.
$n$ is obviously bounded by $\sup_L \sembar{\phi_0}$, which proves termination.
Moreover, as long as $n < \sup_L \sembar{\phi_0}$, there are still words $u$ such that $\sembar{\phi_0}(u) > n$, i.e. $\sembar{\phi}(u) > 0$.
If the search for such words on line~\ref{alg:keypoint} is correct and sound, so is \ComputeBound.
\end{proof}

\subsection{Performance of the algorithm}

Essentially, \ComputeBound enumerates candidate values for $\sup \sembar{\phi_0}$ in increasing order until a fixpoint is reached.
The next candidate is determined on line~\ref{alg:updaten}: the larger $p$, the quicker the algorithm converges.
The choice left for $p$ allows flexibility: the exact value of $\sembar{\phi}(u)$ is most certainly harder to find than an appropriate value $p$.
This line brings a tuning parameter for implementations: the higher the $p$, the faster the convergence, but probably the higher the computation cost.
Implementations should therefore find an appropriate balance between the cost of computing $p$, and the number of loops in \ComputeBound.

To a lesser extent, line~\ref{alg:phi0} brings another tuning parameter for implementations.
$\phi_0$ (resp. $\phi_0[n+1]$) in this line can be safely replaced by $\phi$ (resp. $\phi[n+1]$), without affecting the outcome of the algorithm.
Nevertheless, using the $\phi_0$ variant yields simpler formulae.


\section{Counter $\omega$-Automata}
\label{sec:automata}

This section presents Counter $\omega$-Automata (CA), as introduced by~\cite{bojcol:2006} under the names $B$-automata and $S$-automata.
We also show how to translate \cltlbar{} formulae to CA, based on ideas used for the case of finite words~\cite{kuperberg2014linear}.
We adapt it to infinite words in Section~\ref{sec:cltl_to_automata}, and then show the implementation of our \ComputeBound algorithm with CA in Section~\ref{sec:cegar_implementation}.

Informally, a CA is a $\omega$-automaton equipped with a finite set of non-negative integer counters $\Gamma$, initialized with value $0$.
The values of these counters are controlled by actions: \increment{} that increments a counter; \reset{} that resets a counter to $0$; \observe{} that \emph{observes}, or stores, the current value of the counter.
The set of counter actions is denoted by $\mathds{C}$.
Values of the counters do not affect the behavior of the automaton, but are used to assign a \emph{value} to a word.
Only observed values are used to determine word values.
In addition to a letter $a \in \Sigma$, a CA transition is labelled with $|\Gamma|$ (words of) actions, one for each counter.

\begin{definition}
\label{def:counter_automata}
A counter automaton is a $6$-tuple $\mathcal{A} = \tuple{Q, \Sigma, \Delta, \Gamma, q_0, \mathcal{F}}$ where:
        \begin{compactitem}[-]
          \item $Q$ is a finite set of states, and $q_0 \in Q$ is the initial state;
          \item $\Sigma$ is a finite alphabet;
          \item $\Gamma$ is a finite set of counters;
          \item $\Delta \subseteq Q \times \Sigma \times ({\mathbb{C}^*})^\Gamma \times Q$ is the transition relation;
          \item $\mathcal{F} \subseteq 2^{\Delta}$ is a set of sets of accepting transitions.
        \end{compactitem}  
        An infinite word $u \in \Sigma^\omega$ is accepted by a Counter Automaton $\mathcal{A}$ if there exists an execution of $\mathcal{A}$ on $u$ that visits infinitely often every set in $\mathcal{F}$.
\end{definition}
For $\Gamma = \emptyset$, Definition~\ref{def:counter_automata} defines a $\omega$-automaton.
Along a run $\rho$, counters are incremented and reset according to the encountered actions, and the set of checked values is noted $C(\rho)$.

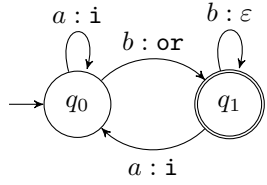
\begin{figure}
  \centering
  \begin{tikzpicture}[node distance=20mm, >=stealth', bend angle=45, auto]
    \node[state, initial, initial text=] (q0) {$q_0$};
    \node[state, accepting] (q1) [right of=q0] {$q_1$};

    \path (q0) edge[loop above] node {$a:\increment{}$} (q0)
          (q0) edge[->, bend left] node {$b:\checkreset{}$} (q1)
          (q1) edge[loop above] node {$b:\varepsilon$} (q1)
          (q1) edge[->, bend left] node {$a:\increment{}$} (q0);
  \end{tikzpicture}
  \caption{A CA counting consecutive $a$'s}
  \label{fig:simple-counter-automaton}
\end{figure}

There are two dual semantics for CA:
\begin{compactitem}
  \item the $\inf$-semantics ($B$-automata in~\cite{bojcol:2006}), where\\
  $\sem{\mathcal{A}}(u) = \inf_{\rho \text{ acc. run on } u} \sup C(\rho)$;
  \item the $\sup$-semantics ($S$-automata in~\cite{bojcol:2006}), where\\
  $\sembar{\mathcal{A}}(u) = \sup_{\rho \text{ acc. run on } u} \inf C(\rho)$.
\end{compactitem}

Figure~\ref{fig:simple-counter-automaton} gives an example of a deterministic CA with the $\sup$-semantics.
Only words in $L = (\Sigma^* b)^{\omega}$ have accepting runs.
Thus $\sembar{\mathcal{A}}(u) = \sup \emptyset = 0$ for $u \notin L$.
If $u \in L$, $\sembar{\mathcal{A}}(u)$ is the smallest size of a block of consecutive $a$'s in $u$.

\subsection{From CLTL to Counter Automata}
\label{sec:cltl_to_automata}

For every \cltl{} (resp. \cltlbar{}) formula $\phi$, there exists a CA $\mathcal{A}_{\phi}$ with the $\inf$-semantics (resp. $\sup$-semantics) with the same semantics: $\sem{\mathcal{A}} = \sem{\phi}$ (resp. $\sembar{\mathcal{A}} = \sembar{\phi}$).
This construction is effective, and does not differ much from the translation from \ltl{} formulae to B\"{u}chi automata (see for instance~\cite{demgas:2012}).
The key difference is the introduction of a counter for each occurrence of the operator $\ltlu{}^{\leq}$ (resp. $\ltlr{}^{>}$) in the formula to translate.
The translation is described in~\cite{kuperberg2014linear} for the case of finite words, and is easily extended to infinite words.

We state here this extension, for the sake of completeness.
In~\cite{kuperberg2014linear}, the produced CA transitions bear \emph{sequences} of counter actions (e.g. a counter can be incremented by three in a single transition).
We show that it is always possible to produce a CA whose transitions are labelled with \emph{atomic actions}, i.e. at most one action (\increment{}, \checkreset{} or $\varepsilon$) per counter.
%
%
This possibility seems to have been overlooked in previous work.
This remark may stem from the care taken in our translation to retain exact values.
We also note that there is a slight difference of semantics for $\cltlbar{}$ with respect to~\cite{kuperberg2014linear}, which is the main cause of the differences between our algorithm and previous ones.
We will also discuss optimizations of the translation.

We label $\ltlr{}^{>}_1$, \dots, $\ltlr{}^{>}_k$ the $k$ occurrences of the operator $\ltlr{}^{>}$ in $\phi$.
Each occurrence is associated a counter, so that $\Gamma = \{ \gamma_1, \dots, \gamma_k \}$.
We note $sub(\phi)$ the set of sub-formulae of $\phi$.

\begin{wraptable}{L}{0.35\textwidth}
\scriptsize
\vspace{-10pt}
  \begin{align*}
    &\text{if } \psi = \psi_1 \land \psi_2 :&&
      {\left\{\begin{array}{l}
        Y \xrightarrow{\varepsilon:\varepsilon} Y \backslash \setof{\psi} \cup \setof{\psi_1, \psi_2}
      \end{array}\right.} \\
    &\text{if } \psi = \psi_1 \lor \psi_2 :&&
      {\left\{\begin{array}{l}
        Y \xrightarrow{\varepsilon:\varepsilon} Y \backslash \setof{\psi} \cup \setof{\psi_1} \\
        Y \xrightarrow{\varepsilon:\varepsilon} Y \backslash \setof{\psi} \cup \setof{\psi_2}
    \end{array}\right.} \\
    &\text{if } \psi = \psi_1 \ltlu{} \psi_2 :&&
      {\left\{\begin{array}{l}
        Y \xrightarrow{\varepsilon:\varepsilon} Y \backslash \setof{\psi} \cup \setof{\psi_2} \\
        Y \xrightarrow[!\psi]{\varepsilon:\varepsilon} Y \backslash \setof{\psi} \cup \setof{\psi_1, \ltlx\psi} \\
    \end{array}\right.} \\
    &\text{if } \psi = \psi_1 \ltlr{} \psi_2 :&&
      {\left\{\begin{array}{l}
        Y \xrightarrow{\varepsilon:\varepsilon} Y \backslash \setof{\psi} \cup \setof{\psi_1, \psi_2} \\
        Y \xrightarrow{\varepsilon:\varepsilon} Y \backslash \setof{\psi} \cup \setof{\psi_2, \ltlx\psi} \\
    \end{array}\right.} \\
    &\text{if } \psi = \psi_1 \ltlr{}^{>}_i \psi_2 :&&
      {\left\{\begin{array}{l}
        Y \xrightarrow{\varepsilon:\checkreset{}_i} Y \backslash \setof{\psi} \cup \setof{\psi_1, \psi_2} \\
        Y \xrightarrow{\varepsilon:\increment{}_i} Y \backslash \setof{\psi} \cup \setof{\psi_1, \psi_2, \ltlx\psi} \\
        Y \xrightarrow{\varepsilon:\varepsilon} Y \backslash \setof{\psi} \cup \setof{\psi_2,\ltlx\psi}
    \end{array}\right.}
  \end{align*}
  \caption{Reduction rules}
  \label{tab:epsilon-reduction}
  \vspace{-10pt}
\end{wraptable}

A state of $\mathcal{A}_{\phi}$ is a set of \cltlbar{} formulae, yet to be verified.
A formula is \emph{reduced} if it is either a literal or its outermost operator is $\ltlx{}$.
A set $Z$ of formulae is \emph{reduced} if it contains only reduced formulae, and \emph{consistent} if it does not contain both a formula and its negation.
Given a reduced and consistent set $Z$, we note $next(Z) = \{ \psi ~|~ \ltlx{} \psi \in Z \}$ and $\Sigma_Z$ the set of letters (in $2^{AP}$) compatible with the literals in $Z$.
$\Sigma_Z$ cannot be empty if $Z$ is consistent.
From a reduced state $Z = \{ l_1, \dots, l_n, \ltlx{} \phi_1, \dots, \ltlx{} \phi_p \}$, reading a letter of $\Sigma_Z$ leads to the state $next(Z) = \{ \phi_1, \dots, \phi_p \}$.

Non-reduced states are reduced step-by-step using $\varepsilon$-transitions, summarized in Table~\ref{tab:epsilon-reduction}, that preserve the state semantics.
Operators $\vee$, $\wedge$, $\ltlu$ and $\ltlr$ follow the classical translation from $\ltl{}$ to $\omega$-automata.
To reduce $\psi = \psi_1 \ltlr{}^{>}_i \psi_2$, three $\varepsilon$-transitions are possible:
\begin{compactitem}
  \item the first one checks and resets the counter $i$, requiring both $\psi_1$ and $\psi_2$ to be verified;
  \item the second one counts one occurrence of $\psi_1$, requiring all $\psi_1$, $\psi_2$ and $\ltlx{} \psi$ to be verified;
  \item the third one does nothing on the counter $i$, and requires both $\psi_2$ and $\ltlx{} \psi$ to be verified.
\end{compactitem}
These three transitions implement in fact the semantics of the operator $\ltlr{}^{>}$.

\begin{figure}
  \centering
  \begin{tikzpicture}[node distance=15mm, >=stealth', auto]
    \tikzstyle{state}=[rounded corners, draw=black!50, fill=white]
    \tikzstyle{pseudo}=[rounded corners, draw=black!50, fill=white, dashed]

    \node[state, initial, initial text=] (q000) {$\neg\phi = \ltlf(p \wedge \ltlg^> \neg q)$};
    \node[pseudo] (q010) [below of=q000] {$p, \ltlg^> \neg q$};
    \node[pseudo] (q020) [below of=q010] {$p, \neg q, \ltlx\ltlg^>\neg q$};
    \node[pseudo] (q011) [right of=q000, xshift=17mm] {$\ltlx\neg\phi$};
    \node[pseudo] (q021) [right of=q010, xshift=17mm] {$p, \neg q$};

    \path (q000) edge[->, dashed] node {$\varepsilon:\varepsilon$} (q010)
          (q010) edge[->, dashed] node {$\varepsilon:\increment{}$} (q020)
          (q010) edge[->, dashed] node {$\varepsilon:\checkreset{}$} (q021)
          (q000) edge[->, dashed] node {$\varepsilon:\varepsilon$} (q011);
  \end{tikzpicture}
  \caption{Reduction of $\setof{\neg\phi}$}
  \label{fig:epsilon}
\end{figure}
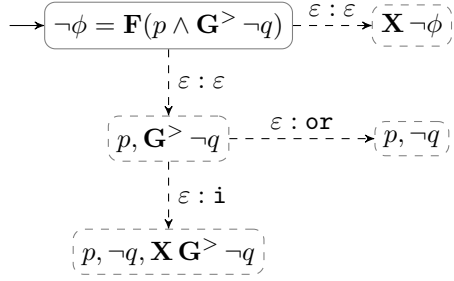

An until formula $\phi_1 \ltlu{} \phi_2$ requires $\phi_2$ to be true at some point.
Transitions subscripted with the label $! \psi$ indicate that $\phi_2$ in $\psi$ has been postponed.
Each until sub-formula in $\phi$ yields one acceptance condition: any transition going through a label $! \psi$ is \emph{not} accepting for the condition $\psi$.
Once the automata with $\varepsilon$-transitions is built, the actual (smaller) automaton is built by collapsing $\varepsilon$-transitions (counter actions are concatenated).

Let us illustrate the described translation with an example: $\phi = \ltlg(p \implies \ltlf^{\leq} q)$.
We turn $\phi$ into an equivalent $\cltlbar{}$ formula $\neg \phi = \ltlf(p \wedge \ltlg^> \neg q)$, which we translate to a CA with $\sup$-semantics.
Figure~\ref{fig:epsilon} depicts the $\varepsilon$-transitions obtained while reducing $\setof{\neg\phi}$.

The reduction yields three reduced sets.
The (not reduced) set $\setof{p \wedge \ltlg^> \neg q}$ is not shown and is directly reduced to $\setof{p, \ltlg^> \neg q}$.
The three reduced sets being also consistent, we are ready to find the real successors of $\neg\phi$, \emph{i.e.} the sets $next(Z)$ where $Z$ is one of the three obtained reduced sets.
First, $next(\ltlx\neg\phi)$ falls back to the initial state $\setof{\neg\phi}$, which will result in a loop in the final automaton.
Note that this will be the only non-accepting transition.
$next(p, \neg q)$ is the set $\setof{\top}$, and $next(p, \neg q, \ltlx\ltlg^> \neg q)$ is $\setof{\ltlg^> \neg q}$.
As this last state is not reduced, the reduction process goes on, yielding $\setof{\neg q}$ and $\setof{\ltlx\ltlg^> \neg q}$.
$next(\neg q) = \setof{\top}$ and $next(\ltlx\ltlg^> \neg q) = \setof{\ltlg^> \neg q}$, states that have already been discovered and reduced.
Finally, we collapse the $\varepsilon$-transitions to get the final automaton shown in Figure~\ref{fig:automaton-final}.

\begin{figure}
\centering
\footnotesize
\subfloat[CA after removing pseudo states]{\label{fig:automaton-final:a}
\begin{tikzpicture}[->,node distance=16mm,auto,scale=0.9,initial text=]
\tikzstyle{every state}=[inner sep=0pt, minimum size=10pt]

\node[initial,state]  (1)                                 {};
\node[state]          (7)   [right=3cm of 1]              {};
\node[state]          (8)   [below right=2cm of 1]        {};

\path[]   (1)   edge [loop above]                               node          {$\top/\varepsilon$}                  (1);
\path[]   (1)   edge [bend left=15]   node[below=-3pt,bacc] {}  node          {$a \wedge \neg b/\checkreset{}$}     (7);
\path[]   (7)   edge [loop right]     node[below=-3pt,bacc] {}  node          {$\top/\varepsilon$}                  (7);
\path[]   (1)   edge [bend left=20]   node[below=-3pt,bacc] {}  node          {$a \wedge \neg b/\varepsilon$}       (8);
\path[]   (1)   edge [bend right=20]  node[below=-3pt,bacc] {}  node [below]  {$a \wedge \neg b/\increment{}~~~~~~~~$}  (8);
\path[]   (8)   edge                  node[below=-3pt,bacc] {}  node [below]  {$~~~~~~\neg b/\checkreset{}$}        (7);
\path[]   (8)   edge [loop right]     node[below=-3pt,bacc] {}  node          {$\neg b/\varepsilon$}                (8);
\path[]   (8)   edge [loop below]     node[below=-3pt,bacc] {}  node          {$\neg b/\increment{}$}               (8);

\end{tikzpicture}
}
\subfloat[CA without unnecessary transitions]{\label{fig:automaton-final:b}
\begin{tikzpicture}[->,node distance=16mm,auto,scale=0.9,initial text=]
\tikzstyle{every state}=[inner sep=0pt, minimum size=10pt]

\node[initial,state]  (1)                                 {};
\node[state]          (7)   [right=3cm of 1]              {};
\node[state]          (8)   [below right=2cm of 1]        {};

\path[]   (1)   edge [loop above]                               node          {$\top/\varepsilon$}                  (1);
\path[]   (1)   edge [bend left=15]   node[below=-3pt,bacc] {}  node          {$a \wedge \neg b/\checkreset{}$}     (7);
\path[]   (7)   edge [loop right]     node[below=-3pt,bacc] {}  node          {$\top/\varepsilon$}                  (7);
\path[]   (1)   edge [bend right=20]  node[below=-3pt,bacc] {}  node [below]  {$a \wedge \neg b/\increment{}~~~~~~~~$}  (8);
\path[]   (8)   edge                  node[below=-3pt,bacc] {}  node [below]  {$~~~~~~\neg b/\checkreset{}$}        (7);
\path[]   (8)   edge [loop below]     node[below=-3pt,bacc] {}  node          {$\neg b/\increment{}$}               (8);
\end{tikzpicture}
}
\caption{CA for $\ltlg(p \implies \ltlf^{\leq} q)$}
\label{fig:automaton-final}
\end{figure}

Several transitions of the automaton of Figure~\ref{fig:automaton-final:a} are unnecessary.
Indeed, according to the $\sup$-semantics, only paths with the higher value are relevant, those with a lower value can safely be removed.
This allows to reduce non-determinism in the automaton, as shown on the automaton of Figure~\ref{fig:automaton-final:b} which has the same semantics as the one of Figure~\ref{fig:automaton-final:a}.

If done appropriately, the actions in the produced automaton can be limited to atomic ones. The proof of Proposition~\ref{prop:atomic_props} is detailed in appendix.
\begin{proposition}
\label{prop:atomic_props}
If the largest (for the sub-formula ordering) formula in $Y$ is picked first when reducing $Y$, then at most one action per counter occurs along any chain of $\varepsilon$-transitions.
\end{proposition}

\subsection{CEGAR-loop Implementation}
\label{sec:cegar_implementation}

We detail the implementation of \ComputeBound using CA, specifically lines~\ref{alg:keypoint} and~\ref{alg:updaten}.
The strength of our algorithm is to boil the problem down to $\omega$-automata emptiness checks, a well-studied problem, with numerous variants and solutions~\cite{rozier2010ltl, renault2013scc}.
This section makes no further assumptions on the variant of $\omega$-automata or the translation algorithm used, so that the final user can use the fittest ones.
Many translations of LTL (and by extension CLTL translation derived from them) to automata produce generalized transition-based automata.

The input language $L$ is assumed to be regular and given as an $\omega$-automaton.
As explained above, we build from $\phi$ a CA $\mathcal{A}_{\phi}$ such that $\sembar{\phi} = \sembar{\mathcal{A}_{\phi}}$.
Line~\ref{alg:keypoint} looks for a word $u \in L$ such that $\sembar{\phi}(u) > 0$.
The constraint $u \in L$ is enforce by searching $u$ such that $\sembar{\mathcal{A}_{\phi} \otimes L}(u) > 0$, where $\mathcal{A}_{\phi} \otimes L$ is the synchronized product of $\mathcal{A}_{\phi}$ and (the automaton of) $L$.
This product, itself a CA, rules out words not in $L$, so that $\sembar{\mathcal{A}_{\phi} \otimes L}(v) > 0$ iff $v \in L$.

Proposition~\ref{prop:cegar_equality} shows that $\sembar{\phi}(u) = 0$ iff there is no $n$ such that $(u,n) \models \phi$.
Given the requirements on $\mathcal{A}_{\phi}$, $\sembar{\phi}(u) = 0$ iff $\mathcal{A}_{\phi}$ has no accepting run on $u$.
Thus, the set of such $u$'s is exactly the language recognized by $\mathcal{A}_{\phi} \otimes L$, viewed as a $\omega$-automaton by ignoring the counters.
Finding such a word $u$ thus amounts to an emptiness check of the said automaton.

A non-empty regular $\omega$-language contains an ultimately periodic word, and so can be chosen $u$, ensuring a finite representation.
In practice, emptiness-check algorithms that compute a counter-example always produce such ultimately periodic words.

Line~\ref{alg:updaten} then asks for a value $p$ between $n$ and $\sembar{\phi_0}(u)$ to update $n$.
We claim that any accepting run $\rho$ on $u$ in the product automaton $\mathcal{A}_{\phi_0} \otimes \mathcal{A}_{\phi_0[n+1]}$ provides such a value $p$.
On the one hand such a $\rho$ is an accepting run in $\mathcal{A}_{\phi_0}$, and its value $p$ is therefore not larger that $\sembar{\phi_0}(u)$.
On the other hand, $\rho$ is also an accepting run in $\mathcal{A}_{\phi_0[n+1]}$.
The whole point of synchronizing $\mathcal{A}_{\phi_0}$ with $\mathcal{A}_{\phi_0[n+1]}$ is to rule out runs of value strictly less than $n$.
Indeed, in $\phi_0[n+1]$, the value $n+1$ is \emph{hard-coded} thanks to $n+1$ nested $\ltlx{}$ operators.
Every time a counter is incremented, a nested $\ltlx{}$ is passed, and $\mathcal{A}_{\phi_0[n+1]}$ accepts a run only if counters are checked with values strictly larger than $n$.
Therefore, replaying $u$ in $\mathcal{A}_{\phi_0} \otimes \mathcal{A}_{\phi_0[n+1]}$ yields a $p$ between $n$ and $\sembar{\phi_0}(u)$.
The great advantage of this operation is that only one run over $u$ needs to be considered, and the computation of $p$ is therefore straightforward.

We recall that higher $p$ speed the convergence of \ComputeBound, by reducing the number of loops.
But higher $p$ would require to explore several runs of $\mathcal{A}_{\phi_0} \otimes \mathcal{A}_{\phi_0[n+1]}$ and to retain the highest found value.
We see more precisely here the trade-off between the number of loops in \ComputeBound and the computation of $p$ on line~\ref{alg:updaten}.

To conclude, we show how \ComputeBound can be extended to also detect the unbounded case, thus providing a complete algorithm.
To this end, we recall that unboundedness of a $\sup$-automaton is decidable, as shown in~\cite{kuperberg2014linear}.
\begin{proposition}{\cite{kuperberg2014linear}}
\label{prop:unbounded}
$\sembar{\mathcal{A}}$ is unbounded if and only if $\mathcal{A}$ has an accepting run $\rho$ in which every action $cr_{\gamma}$ ($\gamma \in \Gamma$) is preceded by a cycle that increments $\gamma$ without resetting $\gamma$.
\end{proposition}
The proof of Proposition~\ref{prop:unbounded} is not difficult: the existence of such a cycle guarantees the ability to build runs with arbitrarily high values.
Conversely, if no such run exists, then every accepting run has its value bounded by $|Q_{\mathcal{A}}|$.

As a corollary of Proposition~\ref{prop:unbounded}, $\sup \sembar{\mathcal{A}}$ is unbounded if and only if $\mathcal{A}$ has an accepting run of value greater than its number of states.
\ComputeBound can thus be adapted so as to detect unboundedness too: compute a bound $B$ on the size of the product $\mathcal{A}_{\phi_0} \otimes \mathcal{A}_L$ (such as $B = |\mathcal{A}_{\phi_0}| \times |\mathcal{A}_L|$).
The sought bound is finite iff $n$ ever exceeds $B$.

%

\section{An example of Application}
  
Through a concrete example, this section illustrates the expressive capabilities of \cltl{} and \cltlbar{} and the kind of problems our bound evaluation algorithm may solve.

\subsection{Ant Colony Optimization}

Ant Colony Optimization\cite{dorigo1997antcolony} (ACO) is a bio-inspired meta-heuristic relying on the cooperative behavior of small simple agents to solve optimization problems.
A collection of artificial \emph{ants} endlessly walk a graph randomly, from some initial node (their \emph{nest}), to one or several target nodes (the \emph{sources of food)}, and come back to the nest.
Whenever an ant moves from one node to another, it deposits a certain amount of pheromone.
The quantity of pheromone left on an edge increases the likelihood that an ant chooses to cross it.
Besides, the quantity of pheromone decreases according to an \emph{evaporation rate}.
Unless the evaporation rate is too high, ants will eventually converge to the shortest paths from their nest to the food sources, because shorter paths will be rewarded with new pheromone more frequently.

%
ACO has been successfully used in numerous applications, such as data mining~\cite{parpinelli02datamining}, image processing~\cite{nezamabadi06edgedetection}.
ACO is resilient to modifications of the graph and it usually responds very quickly to such changes because its current state is likely to contain useful information on the closest new solutions.
Finally, ACO is rather simple to implement on huge distributed setups, as agents do not communicate directly with each others.

\subsection{Quantitative Properties}

Let us consider an ACO that searches the shortest path between two nodes in a directed graph.
The most classic quantitative information is the time (number of steps) taken to find a solution, be it local or global.
Topological parameters may also be measured, such as the maximum length of solutions or the number of nodes visited before a solution is found.
Other quantities relate to the algorithm parameters, such as the maximum amount of pheromones on an edge.
Such information is critical to tune algorithm parameters, that ultimately dictate how fast it converges to a solution\cite{gaertner05onoptimal}.
With a fixed topology, some of these properties are not difficult to compute.
For instance, the minimum number of visited nodes is the length of the shortest path, computable in polynomial time.
Other properties are harder to compute, such as the maximum length of solutions.
When the topology dynamically changes, analytical search for exact optima values is cumbersome, if even possible.

To address these questions, model-checking becomes an option, by checking all possible behaviors of the system.
A common approach instruments the model to monitor the quantitative properties at stake.
It introduces a strong semantical risk, because instrumentation may be impacted by any modification to the model, and must thus be kept up-to-date.
We propose to move the instrumentation into the logics, to keep a proper separation between the actual behavior (the model) and the desired behavior (the logical property).


Let $G=\tuple{V,E}$ be a directed graph where $V$ is a finite set of vertices and $E \subseteq V \times V$ is the set of edges.
An ant is a pair $\tuple{a_v,a_d}$ where $a_v \in V$ is a node, $a_d \in \setof{\Uparrow, \Downarrow}$ is a direction (looking for a food source, and coming back to the nest). $A$ denotes the set of ants.

The time an ant $a$ takes to find a solution is given by
$\phi_{a}(a) = \ltlg (\bot \ltlu{}^{\leq} (a_v = s \land a_d = \Downarrow))$ where $s \in V$ is the nest node.
The worst-case over possible behaviors in $L$ is thus $\sup_L \sem{\phi_{a}}$ and the best-case $\inf_L \sem{\phi_{a}}$.
Similarly, the time taken by the \emph{whole system} to find a solution is obtained by the conjunction over all ants of the previous: $ \phi_{A} = \ltlg \bigwedge_{a \in A} \phi_{a}(a) $.

It is easy to count events like the number of visits of an ant $a$ to a node $s \in V$ with $(a_v = s) \ltlu{}^{\leq} (\ltlg{} \lnot(a_v = s))$.
Occurrences of a position where a \ltl{} formula $\phi$ holds are counted by $\phi \ltlu^{\leq} (\ltlg{} \lnot \phi)$.
Consider the deposit (resp. removal) of a pheromone on edge $e$, denoted by action $\mathtt{add}(e)$ (resp. $\mathtt{rm}(e)$).
The formula $ \phi_{\mathtt{acc}}(e) \equiv \lnot(\mathtt{add}(e) \implies (\lnot\mathtt{add}(e) \ltlu \mathtt{rm}(e))) $ holds in states where $e$ will receive more pheromone before the next removal, i.e. when an ant crosses an edge whose pheromones have not yet evaporated.
The (integer) amount of pheromone on a given edge $e$ is obtained by $\phi_{\mathtt{acc}}(e) \ltlu^{\leq} (\ltlg{} \lnot \phi_{\mathtt{acc}}(e))$.
  
\section{Related work}
\label{sec:related}

A famous problem in language theory is the star-height problem: given a language $L$ (of finite words) and an integer $k$, is there a regular expression for $L$ with at most $k$ nested Kleene stars?
Proposed in 1963~\cite{eggan1963transition}, it was proven decidable in 1988~\cite{hashiguchi1988algorithms} by exhibiting an algorithm with non-elementary complexity, and a much more efficient algorithm was then proposed in 2005~\cite{kirsten2005distance}.
Both algorithms translate the problem to the existence of a bound for a function mapping words to integers, represented in both cases by an automaton equipped with counters (distance automata for the former, nested distance desert automata for the latter).
This boundedness problem of the existence of a bound is then shown decidable. 
It is the first of many problems that reduce to the boundedness problem for such automata.

This motivated an in-depth study of automata with counters (as we use it) as a general framework, that came up with a theory extending the one of regular languages, with logical and algebraic counter-parts~\cite{colcombet2009theory}.
On infinite words, the logical counter-part motivated the introduction and study of \cltl{} and \cltlbar{}~\cite{kuperberg2012expressive}.
This theory also encompasses \emph{promptness} properties, a variant of liveness where a bound on the wait time of a recurring event must exist~\cite{kupferman2007liveness,almagor2010promptness}.
But all these works, motivated by the boundedness problem, overlook the exact values of the functions.
On one hand, this relaxation enables nice closure properties (such as the equivalent expressiveness for $\inf$-automata and $\sup$-automata).
On the other hand, it only allows to reason about the existence of a bound, not to compute values.

In verification, not all questions have a boolean answer, so that various quantitative extensions of automata have been considered, such as weighted automata (see~\cite{droste2007weighted} for a survey).
Despite their various domains of application, they have limited expressivity, as the domain of weights is required to be a semi-ring.
An extension to arbitrary operations on weights have been recently proposed~\cite{alur2013regular}.
It encompasses various extensions of weighted automata, such as Discounted Sum Automata~\cite{alfaro2003discounting} and Counter $\omega$-Automata as considered in this paper.
All these formalisms can be characterized by the absence of guards on register values.
These extensions sometimes have equivalent logics (such as discounted linear temporal logics~\cite{almagor2014discounting}).
From the logical point of view, let us also mention that other temporal logics able to count events were previously proposed~\cite{laroussinie2010counting}.

Most of the cited works only focus on expressivity, decidability and complexity problems, with little consideration to the practical use of such quantitative extensions of automata.
It contrasts with older formalisms: $\omega$-automata have already received great focus towards practical applications, illustrated by numerous emptiness checks algorithms (see~\cite{renault2013scc} for an overview) and many implementations, principally oriented towards \ltl{} model-checking (see~\cite{rozier2010ltl} for a survey).
Some quantitative extensions of automata possess a similar maturity towards practical applications, especially timed automata~\cite{behrmann2006uppaal} and weighted automata~\cite{prism2011}.

\section{Conclusion}
\label{sec:conclusion}

In this paper we proposed to use \cltl{} and \cltlbar{} for practical verification of quantitative properties. One key advantage of these logics is to clearly separate functional properties of the system and quantitative properties, expressed in the logic used for verification. The functional model can still be used for other tasks like production of code and test generation.
Along with examples of properties to be expressed with these logics, we also exhibit a CEGAR-like algorithm to compute bounds for such formulae, based on successive refinements.
We further proposed an implementation of this algorithm using automata equipped with counters, extending the automata approach used for \ltl{} model-checking.

This is a first step towards practical applications of such logics which seems very promising if adequate algorithms and tools are available.
The next step is to implement our algorithm in a proof-of-concept tool.
The logics we used are just a drop in a vast ocean of quantitative extensions for \ltl{}.
Further research should focus on fitting our algorithm in a more general framework so as to capture several such \ltl{} extensions.
Another axis would be the improvement of the performance of validation algorithms: such as improving the translation to automata to produce smaller and/or more deterministic automata and tweak emptiness checks to limit the number of refinement iterations.

\bibliography{concur-2015}

\clearpage
\appendix

\section{Proof of Property~\ref{prop:unfolding}}

\begin{proof}
The proof proceeds by structural induction on $\phi$.
Note that $\phi[n] = \phi$ for any $n$ if $\phi \in \ltl{}$, therefore the property holds on the \ltl{} fragment (and in particular on literals).\\
We recall that $\sem{\phi}(u) \leq n$ iff $(u,n) \modelsinf \phi$.\\
Let us assume that the property holds for \cltl{} formulae $\phi_1$ and $\phi_2$.
Let $\bowtie \in \{ \vee, \wedge, \ltlu{}, \ltlr{} \}$.
$(\phi_1 \bowtie \phi_2)[n] \vdash u$ iff $\phi_1[n] \bowtie \phi_2[n] \vdash u$.
The induction hypothesis allows to replace every occurrence of $\phi_i \vdash v$ in the usual \ltl{} semantics of $\bowtie$ by $(v,n) \modelsinf \phi_i$.
This gives the \cltl{} semantics of $\bowtie$, thus proving that if $(\phi_1 \bowtie \phi_2)[n] \vdash u$, then $(u,n) \modelsinf \phi_1 \bowtie \phi_2$.
The converse reasoning (from \cltl{} to \ltl{} semantics) proves the converse implication.
The same argument is applied to the case $\phi = \ltlx{} \phi_1$.

Now consider $\phi = \phi_1 \ltlu{}^{\leq} \phi_2$.
In the general case, $\phi[n] = (\phi_1[n] \ltlu{}^{\leq} \phi_2[n])[n]$.
Suppose we have a proof for this case when $\phi_1$ and $\phi_2$ are \ltl{} formulae.
The induction hypothesis allows to replace any occurrence of $\phi_1$ and $\phi_2$ in such a proof by $\phi_1[n]$ and $\phi_2[n]$, using the same argument as presented above.
Thus, it suffices to prove the property when $\phi_1$ and $\phi_2$ are \ltl{} formulae to conclude the proof.

Assume $\phi_1$ and $\phi_2$ are \ltl{} formulae.
We now proceed by induction on $n$.
If $n=0$, then $\phi[n] = \phi_1 \ltlu{} \phi_2$.
$u \vdash \phi[n]$ if, and only if, for some index $i$, $u^i \vdash \phi_2$ and $u^j \modelsinf \phi_1$ for every $j < i$.
In other words, $u \vdash \phi[n]$ if, and only if, for some index $i$, $u^i \vdash \phi_2$ and $\{ j < i ~|~ u^j \not\modelsinf \phi_1 \} = \emptyset$.
Thus, $u \vdash \phi[n]$ if, and only, if $(u,0) \modelsinf \phi$.\\
If the property holds at $n$, then $\phi[n+1] = (\phi_1 \vee \ltlx{}(\phi[n])) \ltlu{} \phi_2 = (\phi_1 \ltlu{} \phi_2) \vee (\ltlx{}(\phi[n]) \ltlu{} \phi_2)$.
If $u \vdash (\phi_1 \ltlu{} \phi_2)$, then $(u,0) \modelsinf \phi$ as above, or equivalently $\sem{\phi}(u) \leq 0 < n+1$.
If $u \vdash \ltlx{}(\phi[n]) \ltlu{} \phi_2$, then there is some index $i$ such that $u^i \vdash \phi_2$ and $u^{j+1} \vdash \phi[n]$ for every $j < i$.
Again by induction hypothesis, $(u^i,n+1) \modelsinf \phi_2$.
Let us now consider $S = \{ j < i | (u^j,n+1) \not\modelsinf \phi_1 \}$.
Let $j < i$.
We know that $u^{j+1} \vdash \phi[n]$ which is equivalent, by induction hypothesis, to $(u^{j+1},n) \modelsinf \phi$.
Therefore, there exists an index $i_j > j$ such that $(u^{i_j}, n) \modelsinf \phi_2$.
Since $(u^i,n+1) \modelsinf \phi_2$, we necessarily have $i_j \geq i$.
$S$ is a subset of $T = \{ j < i ~|~ (u^j,n) \not\modelsinf \phi_1 \}$.
We know that $T - \{ 0 \}$ is of size at most $n$ (in case $(u^k,n) \modelsinf \phi$ for every $i \leq k < i_0$).
Therefore, $S$ is of size at most $n+1$, which concludes the proof.
\end{proof}

\section{Proof of Property~\ref{prop:atomic_props}}
\begin{proof}
Consider a path of $\varepsilon$-transitions from a non-reduced state $Y'$ to a reduced state $Y$.
Whenever a formula $\psi$ is reduced along this path, it is removed from the current state ($\ltlx{} \psi$ may appear, but it cannot be reduced until after $Y'$), and only strict sub-formulae of $\psi$ are added.
We claim that the operator $\ltlr{}^{>}_i$ (that occurs only once in $\psi$) is reduced at most once along the path (for every $i$).
Indeed, when $\psi = \psi_1 \ltlr{}^{>}_i \psi_2$ is reduced once, the only way to have it reduced a second time is to be added to the current set by the reduction of another formula $\psi'$. It implies that $\psi$ is a sub-formula of $\psi'$. Since $\psi$ cannot be a sub-formula of one of its strict sub-formulae, then there was a non-reduced formula $\psi''$ with $\psi$ as strict sub-formula when $\psi$ was reduced, which contradicts the selection procedure of formulae to be reduced.
\end{proof}

\end{document}